\documentclass[letterpaper, 10 pt, conference]{ieeeconf}

\usepackage{preamble}
\usepackage{macros}

\title{\bf Bounds of Validity for Bifurcations of Equilibria in a Class of Networked Dynamical Systems}

\author{Pranav Gupta$^1$, Ravi Banavar$^2$, and Anastasia Bizyaeva$^3$%
    \thanks{$^1$ P.G. is with the Centre for Systems and Control, Indian Institute of Technology Bombay, Mumbai - 400076, Maharashtra, India  and with the Sibley School of Mechanical and Aerospace Engineering, Cornell University, Ithaca, NY, USA, 14850 {\tt\small guptapranav@iitb.ac.in}}%
    \thanks{$^2$ R.B. is with the Centre for Systems and Control, Indian Institute of Technology Bombay, Mumbai - 400076, Maharashtra, India {\tt\small banavar@iitb.ac.in}}%
    \thanks{$^3$ A.B. is with the Sibley School of Mechanical and Aerospace Engineering, Cornell University, Ithaca, NY, USA, 14850, {\tt\small anastasiab@cornell.edu}}%
}

\begin{document}

\maketitle
\thispagestyle{empty}
\pagestyle{empty}


\begin{abstract}
    Local bifurcation analysis plays a central role in understanding qualitative transitions in networked nonlinear dynamical systems, including dynamic neural network and opinion dynamics models. In this article we establish explicit bounds of validity for the classification of bifurcation diagrams in two classes of continuous-time networked dynamical systems, analogous in structure to the Hopfield and the Firing Rate dynamic neural network models. Our approach leverages recent advances in computing the bounds for the validity of \LSR{}, a reduction method widely employed in nonlinear systems analysis. Using these bounds we rigorously characterize neighbourhoods around bifurcation points where predictions from reduced-order bifurcation equations remain reliable. We further demonstrate how these bounds can be applied to an illustrative family of nonlinear opinion dynamics on \(k-\)regular graphs, which emerges as a special case of the general framework. These results provide new analytical tools for quantifying the robustness of bifurcation phenomena in dynamics over networked systems and highlight the interplay between network structure and nonlinear dynamical behaviour.
\end{abstract}


\section{Introduction}

Local bifurcations are points of critical transitions in behaviour of nonlinear dynamical systems that occur when a parameter crosses a threshold, causing the stability of an equilibrium point to change and new features such as new equilibria to emerge. Bifurcations are a common route to multistability in complex systems, and bifurcation analysis is used to characterize the emergence, stability, and basins of attraction of coexisting equilibria as parameters vary.

In this paper we study the region of validity for local bifurcation analysis in two broad classes of nonlinear networked dynamical systems, analogous in their structure to Hopfield and firing rate dynamic neural network models. Systems of this form are commonly encountered as models of dynamic neural networks \cite{hopfield1984neurons,betteti2025input,wilson1972excitatory,betteti2025firing}, opinion dynamics and decision-making \cite{franci2015realization,bizyaeva2023nonlinear,gray2018multiagent,fontan2017multiequilibria,abara2017spectral,leonard2024fast}, Lur'e feedback systems \cite{tang2017finite,liu2018bipartite,miranda2018analysis}, gene regulatory networks \cite{goodwin1965oscillatory,rosenfeld2002negative,gardner2000construction}, and other complex systems across application domains. Multistable equilibria are common desirable features in these models, representing meaningful states such as stored associative memories in a neural network or collective decision states in a group of communicating agents.
\LSR{} is a common nonlinear analysis step in the local bifurcation analysis of nonlinear networks to classify such multistabilities \cite{golubitsky2023dynamics,leonard2024fast}. A central challenge is that the method is inherently local, yet the size of the neighbourhood in which the reduced bifurcation equations faithfully capture the local bifurcation diagram of the full system is typically unknown. In this work we will bound these regions of validity for the two broad classes of networked dynamical systems.

In our prior work \cite{gupta2024estimates}, we developed estimates for the domain of validity of reduced-order bifurcation equations for local bifurcations of equilibria derived using \LSR{}. These arise as finite-dimensional equilibrium equations defined on the kernel of the linearized dynamics at the singular point. Our results built directly upon the those established in \cite{jindal2024estimates}, where explicit bounds of validity were obtained for the application of \ImFT{} using tools from functional analysis. In the present work, we build on the general results by specializing the bounds in the context of equilibrium bifurcations in networked dynamical systems. Specifically, we apply this methodology to Hopfield and Firing Rate network models, relating the regions of validity of \LSR{}-based analysis to network structure and choice of smooth nonlinear activation functions in these models.

Our contributions are organized as follows. After introducing notation, relevant lemmas, and key background results on explicit bounds of validity for the \LSR{} in \Cref{sec:background}, we present in \Cref{sec:simple_bounds} a simplified general bound on the regions of validity for \LSR{} in an arbitrary nonlinear dynamical system, building on our previous work in \cite{gupta2024estimates}.
Next, in \Cref{sec:results} we specialize this bound for the Hopfield and Firing Rate network models, presenting simplified computable expressions that leverage the structure of the model equations.
Subsequently, in \Cref{sec:example} we apply the derived bounds to consensus bifurcations in nonlinear opinion dynamics networks over regular graphs. We show that the resulting bounds involve interpretable quantities related to the structure of the underlying graph.
Finally, in \Cref{sec:conclusion} we provide a brief discussion.


\section{Background}
\label{sec:background}

\subsection{Notations and Preliminaries}
\label{subsec:prelims}

Let \(\R\) denote the set of real numbers and \(\N\) the set of positive integers.
A mapping \(f : U \to V\), where \(U \subseteq \R[m]\) and  \(V \subseteq \R[n]\), is said to be of class \(C^\nu\) for \(\nu \in \N\) if it is at least \(\nu\)-times continuously differentiable. For \(f \in C^{\nu}\), we write \(D_{y} f(y)\) for the Jacobian matrix of partial derivatives with respect to the variables \(y\). For \(y_{0} \in \R[n]\) and \(r > 0\), the open ball of radius \(r\) centered at \(y_{0}\) is denoted by \(\mathcal{B}(y_0,r) := \{y \in \R[n] : \norm{y - y_0} < r\}\), where \(\norm{\cdot}\) is a vector norm on \(\R[n]\). For \(A \in \R[p \times q]\), we use \(\norm{A}\) to denote the spectral matrix norm induced by the same vector norm.
For a matrix \(M\), we denote by \(\sigma_{i}(M) = \lambda_i(M^{\top} M)\) its \(i-\)th largest singular value and by \(\sigma_{\min}(M)\) and \(\sigma_{\max}(M)\) its smallest and largest singular values respectively. We let \(\rho(M)\) denote the spectral radius of \(M\).
We write \(\mathbf{0}_{n}, \mathbf{1}_{n} \in \R[n]\) for the all-zero and all-one column vectors, respectively, and \(\mathbf{0}_{p \times q}, \mathbf{1}_{p \times q} \in \R[p \times q]\) for the corresponding matrices.

Let \(\mathcal{G} = (\mathcal{V}, \mathcal{E})\) denote a graph with vertex set \(\mathcal{V}\) of cardinality \(|\mathcal{V}| = n\) and edge set \(\mathcal{E} \subseteq \mathcal{V} \times \mathcal{V}\).
We associate with \(\mathcal{G}\) the signed adjacency matrix \(A \in \R[n \times n]\), where \(a_{ij}\) is the signed weight of the edge between nodes \(i\) and \(j\), and \(a_{ij} = 0\) if no such edge exists.
A graph is said to be \emph{\(k\)-regular} if every vertex has exactly \(k\) neighbours, i.e., the degree of each vertex is equal to \(k\). Equivalently, the adjacency matrix \(A\) of a \(k\)-regular graph satisfies \(A \mathbf{1}_{n} = k \mathbf{1}_{n}\).

\vspace{2mm}

The following lemma provides a convenient matrix norm identity that will be used repeatedly in our analysis.
\begin{lemma}\label{lem:P-norm}
    Let \(U \in \R[n \times k]\) be a matrix with orthonormal columns and let \(P = U U^{\top}\) be the orthogonal projection matrix onto \(\mathrm{col}(U)\). Then for any \(Y \in \R[n \times m]\), we have \(\norm{U^{\top} Y} = \norm{P Y}\) and \(\norm{Y U} = \norm{Y P}\).
\end{lemma}
\begin{proof}
    Since \(U^{\top}U = I\), and \(P^2=P\), from the definition of spectral radius we have
    \(
    \norm{U^{\top}Y}^2 = \rho(Y^{\top}UU^{\top}Y)
    = \rho(Y^{\top}PY)
    = \rho(Y^{\top}P^2Y)
    = \norm{PY}^2.
    \)
    Using \(P^{\top}=P\), we also have
    \(
    \norm{YU} = \norm{U^{\top}Y^{\top}}
    = \norm{PY^{\top}}
    = \norm{YP}.
    \)
\end{proof}


\subsection{\LSR{}}
\label{subsec:LSR}

\LSR{} is a method used to analyse parametrized nonlinear dynamical systems near their singular equilibria, leveraging the \ImFT{}. The method decomposes the state space of a dynamical system into the null space of its Jacobian and its orthogonal complement at the singular point. The \ImFT{} is then used to eliminate the component orthogonal to the null space in the local bifurcation diagram of the system, effectively projecting the high-dimensional equations describing the equilibria of the model onto the kernel \cite[Ch. VII]{Golubitsky1985}. The resulting reduced equations capture the structure of the local bifurcation diagram of the original system. We now present a partial summary of the reduction procedure.

Let \(X \subset \R[n]\) and \(\Lambda \subset \R[m]\) be open sets. Consider the parametrized nonlinear dynamical system \(\dot{x} = \Phi(x,p)\) where \(x \in X\) is the state vector, \(p \in \Lambda\) is the parameter vector, and \(\Phi : X \times \Lambda \to \R[n]\) is at least \(\mathcal{C}^{2}\), with a bifurcation at some singular point \((x^{\star},p^{\star}) \in X \times \Lambda\), i.e. \(\Phi(x^{\star},p^{\star}) = 0\) and \(J^{\star} := D_{x}\Phi(x^{\star},p^{\star})\) is singular (non-invertible) with \(q := \dim\bigl(\kerJ\bigr) > 0\).

Let us denote the SVD of \(J^{\star}\) in the following manner:
\begin{equation}
    J^{\star} = \begin{bmatrix}W & \Wbar\end{bmatrix} \begin{bmatrix} \Sigma & \mathbf{0}_{q \times (n-q)} \\ \mathbf{0}_{(n-q) \times q} & \mathbf{0}_{q \times q}\end{bmatrix} \begin{bmatrix}\Vbar^{\top} \\ V^{\top}\end{bmatrix} = W\Sigma\Vbar^{\top} \label{eq:SVD}
\end{equation}
with \(\Sigma = \diag(\sigma_{1},\dots,\sigma_{n-q})\) and \(\sigma_{1} \geq \dots \geq \sigma_{n-q} > 0\). Naturally \(V, \Wbar \in \R[n \times q]\) correspond to the null singular values and \(\Vbar, W \in \R[n \times (n-q)]\) correspond to the non-zero singular values. By the properties of the SVD, the columns of \(V\) and \(W\) form orthonormal bases for the null and range spaces of \(J^{\star}\), while those of \(\Vbar\) and \(\Wbar\) form orthonormal bases their respective orthogonal complement subspaces.
Consequently, the corresponding orthogonal projections are
\begin{align}
    P_{\kerJ}                     & = VV^{\top},                                  & \quad
    P_{\range(J^{\star})}         & = WW^{\top}, \notag                                   \\
    P_{\kerJ^{\perp}}             & = \Vbar\,\Vbar^{\top},                        & \quad
    P_{\range(J^{\star})^{\perp}} & = \Wbar\,\Wbar^{\top}. \label{eq:projections}
\end{align}

Using these, the unique orthogonal decomposition of the state \(x\in\R[n]\) into \(\kerJ \oplus \kerJ^{\perp}\) can be denoted by
\begin{equation*}
    X \ni x \mapsto (x^{\parallel},x^{\perp}) := (VV^{\top} x, \Vbar\,\Vbar^{\top} x) \in \kerJ \times \kerJ^{\perp}.
\end{equation*}
Define coordinate map \(\R[q] \times \R[n-q] \ni (\alpha,\beta) \mapsto \Gamma(\alpha,\beta) \in X\):
\begin{equation}
    x = \Gamma(\alpha,\beta) := \underbrace{V\alpha}_{x^\parallel} + \underbrace{\Vbar\beta}_{x^\perp} \quad\text{where}\quad \begin{cases} \alpha = V^{\top}x \in \R[q] \\ \beta = \Vbar^{\top}x \in \R[n-q] \end{cases} \label{def:split}
\end{equation}
where \(\alpha\) and \(\beta\) are the coordinates of \(x^{\parallel}\) and \(x^{\perp}\) in the bases \(V\) and \(V^{\top}\), respectively. We employ this coordinate map as a shorthand throughout the analysis.

The \LSR{} involves applying the \ImFT{} to obtain a locally valid implicit mapping \((\alpha,p) \mapsto \beta\). In particular, it guarantees the existence of open neighbourhoods \(\mathcal{N}\big((\alpha_0,p^{\star})\big) \subset \R[q] \times \Lambda\) and \(\mathcal{N}(\beta_0) \subset \R[n-q]\) equipped with implicit map \(\varphi:\mathcal{N}\big((\alpha_0,p^{\star})\big) \to \mathcal{N}(\beta_0) \)
that expresses coordinates of the state \(x\) along the \(\kerJ^{\perp}\) as a function of its coordinates along \(\kerJ\) and the bifurcation parameter, i.e. \(\beta = \varphi(\alpha,p)\), satisfying
\begin{align}
    \Phi(V\alpha+\Vbar\varphi(\alpha,p),p) & = 0\quad\forall\;(\alpha,p)\in\mathcal{N}\big((\alpha^{\star},p^{\star})\big). \label{eq:LSRmap}
\end{align}


\subsection{Bounds of Validity for \LSR{} \label{subsec:LSR-bounds}}

A key limitation of \LSR{} is that the equivalence between the reduced bifurcation equations and the equilibria of the original system holds only locally, restricting the validity of analysis based on this reduction to a neighbourhood of the singular point. Quantifying this neighbourhood is crucial for evaluating the robustness of conclusions drawn from the reduced equations. To address this limitation, in previous work we derived a general bound of validity for \LSR{} for bifurcations of equilibria in systems of ordinary differential equations \cite{gupta2024estimates}. In this section we summarize our recent result.

In the context of \LSR{}, we derive explicit lower bounds on the neighbourhoods of the implicit map \(\varphi\) as open balls, i.e. \(\mathcal{B}\big((\alpha_0,p^{\star}),\rpr\big) \subset \mathcal{N}\big((\alpha_0,p^{\star})\big)\) and \(\mathcal{B}(\beta_0,\rpp) \subset \mathcal{N}(\beta_0)\). The radii \(\rpr\) and \(\rpp\) quantify the neighbourhood of the singular point \((\alpha,\beta,p) = (\alpha_0,\beta_0,p^{\star})\) within which the reduced and full bifurcation diagrams coincide topologically: \(\rpr\) specifies the radius along the critical subspace \(\kerJ \times \Lambda\), while \(\rpp\) specifies the corresponding radius along \(\kerJ^{\perp}\). The algorithm to explicitly determine the bounding radii \(\rpr\) and \(\rpp\) is described below.

Before we state the theorem on bounds, we define the following quantities.
First, we define two scalar quantities that will appear in the statement of the bounds,
\begin{subequations}
    \begin{align}
        \Mpr & := \norm{\begin{bmatrix}\mathbf{0}_{(n-q) \times q} & W^{\top} D_{p}\Phi(x^{\star},p^{\star})\end{bmatrix}}, \label{def:M_para} \\
        \Mpp & := \norm{\left(W^{\top} J^{\star} \Vbar\right)^{-1}}. \label{def:M_perp}
    \end{align} \label{eq:Ms}%
\end{subequations}
Next, we define the maps \(\xi_{\parallel}: \R[q] \times \Lambda \to \R[(n-q)\times(q+m)]\) and \(\xi_{\perp}: \R[q] \times \R[n-q] \times \Lambda \to \R[(n-q)\times(n-q)]\) as
\begin{subequations}
    \begin{align}
        \xipr{p} & := W^{\top}\begin{bmatrix} \xi_{\parallel\alpha}(\alpha,p) & \xi_{\parallel p}(\alpha,p) \end{bmatrix}, \label{def:xi_para} \\
        \xipp{p} & := W^{\top}\bigl(D_{x}\Phi\bigl(\Gamma(\alpha,\beta),p\bigr) - J^{\star}\bigr) W ,\label{def:xi_perp}
    \end{align} \label{eq:xis}%
\end{subequations}
where \(\xi_{\parallel p}(\alpha,p) = D_{p}\Phi\bigl(\Gamma(\alpha,\beta_0),p\bigr)-D_{p}\Phi(x^{\star},p^{\star})\) and \(\xi_{\parallel\alpha}(\alpha,p) = D_{x}\Phi\bigl(\Gamma(\alpha,\beta_0),p\bigr)V\).
The maps \eqref{eq:xis} quantify first-order variations of the Jacobian \(D_{x}\Phi\) and parameter derivative \(D_{p}\Phi\) along the critical subspace and its orthogonal complement.
We use \(\Bpr := \mathcal{B}\big((\alpha_0,p^{\star}),\rpr\big) \subset \R[q] \times \Lambda \) and \(\Bpp := \mathcal{B}\big(\beta_0,\rpp\big) \subset \R[n-q]\) as shorthand notation to denote open balls centred at the critical point.
Finally, we define two scalar quantities as the maximum norm of \(\xipr{p}, \xipp{p}\) over open balls centred around the critical point, and state the main result from \cite{gupta2024estimates}
\begin{subequations}
    \begin{align}
        \Lpr & := \sup_{(\alpha,p) \in \Bpr} \norm{\xipr{p}} \label{def:L_para} \\
        \Lpp & := \sup_{\substack{(\alpha,p) \in \Bpr                           \\ \beta \in \Bpp}} \norm{\xipp{p}}. \label{def:L_perp}
    \end{align}
\end{subequations}

\begin{theorem}\label{thm:LSR-bounds}  \cite[Theorem 3.2]{gupta2024estimates}
    For all \(\rpr,\rpp>0\) simultaneously satisfying
    \begin{subequations}
        \begin{align}
            \Lpr\rpr + \Lpp\rpp & < \frac{\rpp}{\Mpp} - \Mpr\rpr \label{eq:LSR_1} \\
            \Mpp\Lpp            & < 1 \label{eq:LSR_2}
        \end{align} \label{eq:LSR-bounds}%
    \end{subequations}
    there exists a smooth implicit map \(\varphi : \Bpr \to \Bpp\) that satisfies \eqref{eq:LSRmap}, where the relevant quantities in \eqref{eq:LSR-bounds} are defined by the equations \eqref{def:M_para}, \eqref{def:M_perp}, \eqref{def:L_para}, \eqref{def:L_perp}.
\end{theorem}

\section{Simplified General LS Bounds}
\label{sec:simple_bounds}

Before stating our main results in the next section, we present some results simplifying the statement of the general bounds of \Cref{thm:LSR-bounds}.
First we streamline the expressions \eqref{def:M_para}, \eqref{def:M_perp} in terms of the projections \(P_{\range(J^{\star})}=WW^{\top}\) and the partial derivatives of \(\Phi\) at \((x^{\star},p^{\star})\) w.r.t. \(x\) and \(p\).

\vspace{3mm}

\begin{lemma}\label{lem:M_para}
    \(\Mpr = \norm{P_{\range(J^{\star})} D_{p}\Phi(x^{\star},p^{\star})}.\)
\end{lemma}

\begin{proof}
    From \eqref{def:M_para} and definition of spectral radius,
    \(
    \Mpr^2 = \\ \rho\left(
    \begin{bmatrix} \mathbf{0}_{n-q} \!&\! W^{\top} J_{p}^{\star} \end{bmatrix}^{\top}
    \begin{bmatrix} \mathbf{0}_{n-q} \!&\! W^{\top} J_{p}^{\star} \end{bmatrix}
    \right) = \rho\bigl(J_{p}^{\star\top} W W^{\top} J_{p}^{\star}\bigr) \\
    = \norm{W^{\top} J_{p}^{\star}}^2,
    \) where \(J_{p}^{\star}:=D_{p}\Phi(x^{\star},p^{\star})\).
    Taking square root and applying \Cref{lem:P-norm} completes the proof.
\end{proof}

\vspace{4mm}

\begin{lemma}\label{lem:M_perp}
    \(\Mpp = 1/\sigma_{\min}\left(J^{\star}\right).\)
\end{lemma}

\begin{proof}
    From its definition in \eqref{def:M_perp},
    \(
    \Mpp = \norm{\bigl(W^{\top} J^{\star} \Vbar\bigr)^{-1}}
    = \norm{\bigl(W^{\top} W\Sigma\Vbar^{\top} \Vbar\bigr)^{-1}}
    = \norm{\Sigma^{-1}} = \sigma_{\max}(\Sigma^{-1}) = \frac{1}{\sigma_{\min}(\Sigma)} \\
    = 1/\sigma_{\min}(J^{\star}),
    \)
    since \(\Sigma\) is the diagonal matrix of singular values (all strictly positive) of \(J^{\star}\).
\end{proof}

Next, we prove that the second inequality condition in the originally derived \Cref{thm:LSR-bounds} is redundant, and state a simplified version of the bounds.
\begin{theorem} \label{thm:LSR-bounds-simple}
    For all \(\rpr,\rpp>0\) satisfying
    \begin{equation}
        \Lpr\rpr + \Lpp\rpp  < \frac{\rpp}{\Mpp} - \Mpr\rpr \label{eq:LSR-bounds-simple}
    \end{equation}
    there exists a smooth implicit map \(\varphi : \Bpr \to \Bpp\) that satisfies \eqref{eq:LSRmap}, where the quantities in \eqref{eq:LSR-bounds-simple} are defined by \Cref{lem:M_para}, \Cref{lem:M_perp}, and equations \eqref{def:L_para}, \eqref{def:L_perp}.
\end{theorem}
\begin{proof}
    Suppose \eqref{eq:LSR_1} holds. Rearranging terms in this inequality, we can equivalently state that
    \begin{equation*}
        \Mpp \Lpp + \Mpp \big(\Mpr + \Lpr\big) \frac{\rpr}{\rpp} < 1.
    \end{equation*}
    Since all the quantities on the above left hand side are non-negative scalars, it follows that \(\Mpp \Lpp < 1\) must necessarily hold, i.e. \eqref{eq:LSR_2} is satisfied. The theorem statement then follows directly from \Cref{thm:LSR-bounds}.
\end{proof}


\section{Main Results: Network Bounds}
\label{sec:results}

In this section we study the robustness of bifurcation diagrams for two general classes of networked nonlinear dynamical systems.
First we consider systems of the form
\begin{equation}
    \dot{x} = \Phi(x,p) := -Cx + pAS(x) + b \label{eq:hopfield}
\end{equation}
where \(x \in \R[n]\) is the state vector, \(p \in \R[]\) is a parameter, \(b \in \R[n]\) is a constant input or bias, \(A \in \R[n \times n]\), and \(C=\diag(c_1,\dots,c_n)\) with \(c_i \in \R[>0]\). The function \(S(z) = \begin{bmatrix} S_1(z_1) & \dots & S_n(z_n) \end{bmatrix}^{\top}\) applies an element-wise smooth activation function to each input, and its derivative is denoted as \(D_{z}S(z) = \diag\bigl(S'_1(z_1), \dots, S'_n(z_n)\bigr)\). We will refer to \eqref{eq:hopfield} as a \emph{Hopfield} or \emph{voltage} model as it follows the structure of classic and modern associative memory networks \cite{hopfield1984neurons,betteti2025input}.
The second class we will consider are systems of the form
\begin{equation}
    \dot{x} = \Psi(x,p) := -Cx + S(pAx + b) \label{eq:firing}
\end{equation}
with all the quantities defined the same as for the Hopfield models. We will refer to \eqref{eq:firing} as a \emph{firing rate} model \cite{wilson1972excitatory,betteti2025firing}. It has been shown that \eqref{eq:hopfield}, \eqref{eq:firing} are related through changes of coordinates and time-varying input transformations \cite{miller2012mathematical}.

In this section we establish computable bounds of validity for the equivalence of bifurcation diagrams of \eqref{eq:hopfield}, \eqref{eq:firing} with their reduced bifurcation equations obtained via \LSR{} at singular points, specializing \Cref{thm:LSR-bounds-simple}. We seek to state these bounds in terms of interpretable quantities such as the network connectivity matrix \(A\), the activation function \(S\), and its derivatives.


\subsection{Hopfield models}
\label{subsec:hopfield}

First we consider \eqref{eq:hopfield}. Computing derivatives of the model with respect to \(x\) and \(p\) yields
\begin{subequations}
    \begin{align}
        D_{x}\Phi(x,p) & = -C + pAD_{z}S(x), \label{eq:hopfield-D_x} \\
        D_{p}\Phi(x,p) & = AS(x). \label{eq:hopfield-D_p}
    \end{align} \label{eq:hopfield-jacobians}
\end{subequations}
We define the Jacobian \(J^{\star}_H\) at the singular point as
\begin{equation*}
    J^{\star}_H := D_{x}\Phi(x^{\star},p^{\star}) = -C + p^{\star} A D_{z}S(x^{\star})
\end{equation*}
and projection matrices \(P:=P_{\range(J^{\star}_H)}=W_H W_H^{\top}\) and \(Q:=P_{\ker(J^{\star}_H)^{\perp}}=\Vbar_H \Vbar_H^{\top}\) from \(J^{\star}_H\) as described in \eqref{eq:SVD} and \eqref{eq:projections}. Next, we make the following assumptions:
\begin{assumption}\label{asmp:hopfield}
    There exists \((x^{\star},p^{\star}) \in \R[n] \times \R[>0]\) such that the following conditions hold:
    \begin{enumerate}[label=(\roman*)]
        \item The pair \((x^{\star},p^{\star})\) is an equilibrium of \eqref{eq:hopfield}.
        \item The Jacobian \(J^{\star}_H\) has a zero eigenvalue with multiplicity \(q\), and \(\Re(\lambda_{j}) \neq 0\) for all other eigenvalues.
    \end{enumerate}
\end{assumption}
To apply \Cref{thm:LSR-bounds-simple} to \eqref{eq:hopfield}, we compute the various quantities in the inequality condition \eqref{eq:LSR-bounds-simple}.
In the following sequence of lemmas we derive expressions needed to express \(\Lpr,\Lpp\) in \eqref{eq:LSR-bounds-simple}, exploiting the form of Hopfield dynamics to express these quantities in terms of the network matrix \(A\), activation function \(S\), and its first derivatives.

\vspace{3mm}

\begin{lemma}\label{lem:hopfield-xi_para}
    \(\norm{\xipr{p}} = \norm{P A \begin{bmatrix}\Delta_{\parallel}(\alpha,p) V & \bar{S}_{\parallel}(\alpha)\end{bmatrix}},\)\\
    where \(\Delta_{\parallel}(\alpha,p)\) and \(\bar{S}_{\parallel}(\alpha)\) are defined as %
    \vspace{0mm} \begin{subequations}
        \begin{align} \label{eq:hopfield-xi_para_terms}
            \Delta_{\parallel}(\alpha,p) & := p D_{z}S\bigl(\Gamma(\alpha,\beta_0)\bigr) - p^{\star} D_{z}S(x^{\star}) \\
            \bar{S}_{\parallel}(\alpha)  & := S\bigl(\Gamma(\alpha,\beta_0)\bigr) - S(x^{\star})
        \end{align}
        \vspace{-3mm} \end{subequations}
\end{lemma}
\begin{proof}
    Following the definition \eqref{def:xi_para}, we must compute \(\xi_{\parallel\alpha}(\alpha,p)\) and \(\xi_{\parallel p}(\alpha,p)\).
    Since \(J^{\star}_H V = 0\), we can write \\
    \(
    \bigl(-C + p^{\star} A D_{z}S(x^{\star})\bigr) V = 0 \implies C V = p^{\star} A D_{z}S(x^{\star}) V,
    \)
    which allows us to simplify \(\xi_{\parallel\alpha}(\alpha,p)\) as
    \begin{align*}
        \xi_{\parallel\alpha}(\alpha,p) & = D_{x}\Phi\bigl(\Gamma(\alpha,\beta_0),p\bigr)V                                             \\
                                        & = \bigl[-C+p A D_{z}S\bigl(\Gamma(\alpha,\beta_0)\bigr)\bigr]V                               \\
                                        & = \bigl[-p^{\star} A D_{z}S(x^{\star}) + p A D_{z}S\bigl(\Gamma(\alpha,\beta_0)\bigr)\bigr]V \\
                                        & =: A \Delta_{\parallel}(\alpha,p) V
    \end{align*}
    Next, we express
    \(
    \xi_{\parallel p}(\alpha,p) = D_{p}\Phi\bigl(\Gamma(\alpha,\beta_0),p\bigr) - \\ D_{p}\Phi(x^{\star},p^{\star}) = A \bigl(S\bigl(\Gamma(\alpha,\beta_0)\bigr) - S(x^{\star})\bigr) = A \bar{S}_{\parallel}(\alpha).
    \)

    Finally, combining our expressions for \(\xi_{\parallel\alpha}\) and \(\xi_{\parallel p}\) yields
    \(
    \xipr{p} = W^{\top} A \begin{bmatrix} \Delta_{\parallel}(\alpha,p) V & \bar{S}_{\parallel}(\alpha) \end{bmatrix}.
    \) Applying \Cref{lem:P-norm} to this yields the desired expression.
\end{proof}

\vspace{3mm}

\begin{lemma}\label{lem:hopfield-xi_perp}
    \(\norm{\xipp{p}} = \norm{P A \Delta_{\perp}(\alpha,\beta,p) Q}\) where
    \vspace{0mm} \begin{equation}
        \Delta_{\perp}(\alpha, \beta, p) := p D_{z}S\bigl(\Gamma(\alpha,\beta)\bigr) - p^{\star} D_{z}S(x^{\star}). \label{eq:hopfield-xi_perp_terms}
        \vspace{2mm} \end{equation}
\end{lemma}
\begin{proof}
    Building from its definition \eqref{def:xi_perp}, we have
    \begin{align*}
        \xipp{p} & = W^{\top}\Bigl(D_{x}\Phi\bigl(\Gamma(\alpha,\beta),p\bigr)-D_{x}\Phi(x^{\star},p^{\star})\Bigr)\Vbar \\
                 & = W^{\top}\Bigl(p A D_{z}S\bigl(\Gamma(\alpha,\beta)\bigr) - p^{\star} A D_{z}S(x^{\star})\Bigr)\Vbar \\
                 & =: W^{\top} A \Delta_{\perp}(\alpha, \beta, p) \Vbar.
    \end{align*}
    The desired expression follows directly by applying \Cref{lem:P-norm} as follows:
    \(
    \norm{W^{\top} A \Delta_{\perp}(\alpha, \beta, p) \Vbar} = \\
    \norm{WW^{\top} A \Delta_{\perp}(\alpha, \beta, p) \Vbar\,\Vbar^{\top}} = \norm{P A \Delta_{\perp}(\alpha, \beta, p) Q}.
    \)
\end{proof}

\vspace{2mm}

With these lemmas in place, we now state the bounds from \Cref{thm:LSR-bounds-simple} for bifurcations in the Hopfield model.

\vspace{3mm}

\begin{theorem}\label{thm:hopfield-bounds}
    Consider the Hopfield dynamics \eqref{eq:hopfield} under \Cref{asmp:hopfield}. For all \(\rpr,\rpp>0\) satisfying
    \begin{multline}
        \Lpr\rpr + \Lpp\rpp \\ < \sigma_{\min}\left(J^{\star}_H\right) \rpp - \tfrac{1}{\abs{p^{\star}}} \norm{P(Cx^{\star} - b)} \rpr
    \end{multline}
    where the quantities \(\Lpr,\Lpp\) are defined as
    \begin{align*}
        \Lpr & = \underset{(\alpha,p) \in \Bpr}{\sup} \norm{P A \begin{bmatrix}\Delta_{\parallel}(\alpha,p)V & \bar{S}_{\parallel}(\alpha) \end{bmatrix}} \\
        \Lpp & = \underset{\substack{(\alpha,p) \in \Bpr                                                                                                  \\ \beta \in \Bpp}}{\sup} \norm{P A \Delta_{\perp}(\alpha,\beta,p) Q},
    \end{align*}
    with \(\Delta_{\parallel}\), \(\Delta_{\perp}\), \(\bar{S}_{\parallel}\) defined in \eqref{eq:hopfield-xi_para_terms}, \eqref{eq:hopfield-xi_perp_terms}, there exists a smooth implicit map \(\varphi : \Bpr \to \Bpp\) satisfying \eqref{eq:LSRmap}.
\end{theorem}
\begin{proof}
    From the equilibrium condition, we see that \(AS(x^{\star}) = \frac{1}{p^{\star}} (Cx^{\star} - b)\). The statement then immediately follows from applying \Cref{lem:M_para} to obtain
    \begin{equation}
        \Mpr = \norm{P A S(x^{\star})} = \tfrac{1}{\abs{p^{\star}}}\norm{P(Cx^{\star} - b)}. \label{eq:M_para}
    \end{equation}
    Substituting \eqref{def:M_perp}, \eqref{def:L_para}, \eqref{def:L_perp} from Lemmas \ref{lem:M_perp}, \ref{lem:hopfield-xi_para}, \ref{lem:hopfield-xi_perp} into \eqref{eq:LSR-bounds-simple} from \Cref{thm:LSR-bounds-simple} yields the theorem statement.
\end{proof}


\subsection{Firing Rate models}
\label{subsec:firing}

Next, we repeat the presented analysis for \eqref{eq:firing}.
The derivatives of \eqref{eq:firing} with respect to \(x\) and \(p\) are
\begin{subequations}
    \begin{align}
        D_{x}\Psi(x,p) & = -C + pD_{z}S(pAx + b)A \label{eq:firing-D_x} \\
        D_{p}\Psi(x,p) & = D_{z}S(pAx + b)Ax \label{eq:firing-D_p}
    \end{align} \label{eq:firing-jacobians}
\end{subequations}
and the Jacobian \(J^{\star}_F\) at the singular point evaluates to
\begin{equation*}
    J^{\star}_F := D_{x}\Psi(x^{\star},p^{\star}) = -C + p^{\star} D_{z}S(p^{\star} A x^{\star} + b)A.
\end{equation*}
Define projection matrices \(P:=P_{\range(J^{\star}_F)}=W_F W_F^{\top}\) and \(Q:=P_{\ker(J^{\star}_F)^{\perp}}=\Vbar_F \Vbar_F^{\top}\) as described in \eqref{eq:SVD} and \eqref{eq:projections}.  We make the following assumption:
\begin{assumption}\label{asmp:firing}
    There exists \((x^{\star},p^{\star}) \in \R[n] \times \R[>0]\) such that the following conditions hold:
    \begin{enumerate}[label=(\roman*)]
        \item The pair \((x^{\star},p^{\star})\) is an equilibrium of \eqref{eq:firing}.
        \item The Jacobian \(J^{\star}_F\) has a zero eigenvalue with multiplicity \(q\), and \(\Re(\lambda_{j}) \neq 0\) for all other eigenvalues.
    \end{enumerate}
\end{assumption}
Analogously to the computation carried out for Hopfield networks, in the following sequence of Lemmas we derive expressions needed to define \(\Lpr,\Lpp\) in order to apply \Cref{thm:LSR-bounds-simple} to \eqref{eq:firing}.

\vspace{3mm}

\begin{lemma}\label{lem:firing-xi_para}
    \(\norm{\xipr{p}} = \norm{P \begin{bmatrix} \Delta^{\parallel}_{x}(\alpha,p) AV & \Delta^{\parallel}_{p}(\alpha,p) \end{bmatrix}}\) where \(\Delta^{\parallel}_{x}(\alpha,p)\) and \(\Delta^{\parallel}_{p}(\alpha,p)\) are defined as
    \vspace{-1mm}\begin{subequations} \label{eq:firing-delta_para}
        \begin{align}
             & \Delta^{\parallel}_{x}(\alpha,p) := p D_{z}S\bigl(p A \Gamma(\alpha,\beta_0) + b\bigr) - p^{\star}D_{z}S(p^{\star} A x^{\star} + b) \\
             & \Delta^{\parallel}_{p}(\alpha,p) := D_{z}S\bigl(p A \Gamma(\alpha,\beta_0) + b\bigr)A\Gamma(\alpha,\beta_0) \notag                  \\
             & \hspace{22mm} - D_{z}S(p^{\star} A x^{\star} + b) A x^{\star}
        \end{align}
        \vspace{-5mm}\end{subequations}
\end{lemma}
\begin{proof}
    Following the definition \eqref{def:xi_para}, we must compute \(\xi_{\parallel\alpha}(\alpha,p)\) and \(\xi_{\parallel p}(\alpha,p)\).
    Since \(J^{\star}_F V = 0\) we can write \(CV = p^{\star}D_{z}S(p^{\star} A x^{\star} + b)AV,\)
    simplifying \(\xi_{\parallel\alpha}(\alpha,p)\) as
    \begin{align*}
        \xi_{\parallel\alpha}(\alpha,p) & = D_{x}\Psi\bigl(\Gamma(\alpha,\beta_0),p\bigr)V                                                        \\
                                        & = \bigl[-C + p D_{z}S\bigl(p A \Gamma(\alpha,\beta_0) + b\bigr)A\bigr]V                                 \\
                                        & = \bigl[p D_{z}S\bigl(p A \Gamma(\alpha,\beta_0) + b\bigr) \notag                                       \\
                                        & \hspace{2mm} - p^{\star}D_{z}S(p^{\star} A x^{\star} + b)\bigr]AV =: \Delta^{\parallel}_{x}(\alpha,p)AV
    \end{align*}
    Next, we express
    \begin{align*}
        \xi_{\parallel p}(\alpha,p) & = D_{p}\Psi(\Gamma(\alpha,\beta_0),p)-D_{p}\Psi(x^{\star},p^{\star})                                    \\
                                    & = \bigl[D_{z}S\bigl(p A \Gamma(\alpha,\beta_0) + b\bigr)A\Gamma(\alpha,\beta_0) \notag                  \\
                                    & \hspace{10mm} - D_{z}S(p^{\star} A x^{\star} + b) A x^{\star}\bigr] =: \Delta^{\parallel}_{p}(\alpha,p)
    \end{align*}
    Finally, combining our expressions for \(\xi_{\parallel\alpha}\) and \(\xi_{\parallel p}\) yields
    \(
    \xipr{p} = W^{\top} \begin{bmatrix} \Delta^{\parallel}_{x}(\alpha,p) AV & \Delta^{\parallel}_{p}(\alpha,p) \end{bmatrix}.
    \)
    Applying \Cref{lem:P-norm} to this yields the desired expression.
\end{proof}

\vspace{3mm}

\begin{lemma}\label{lem:firing-xi_perp}
    \(\norm{\xipp{p}} = \norm{P \Delta_{\perp}(\alpha,\beta,p) A Q}\) where
    \vspace{0mm} \begin{equation}
        \Delta_{\perp}(\alpha,\beta,p) := pD_{z}S(pAx + b) - p^{\star}D_{z}S(p^{\star}Ax^{\star} + b) \label{eq:firing-delta_perp}
        \vspace{2mm} \end{equation}
\end{lemma}
\begin{proof}
    Building from its definition \eqref{def:xi_perp}, we have
    \begin{align*}
        \xipp{p} & = W^{\top}\big(D_{x}\Psi(\Gamma(\alpha,\beta),p)-D_{x}\Phi(x^{\star},p^{\star})\big)\Vbar \\
                 & = W^{\top}\bigl(pD_{z}S\bigl(pA\Gamma(\alpha,\beta) + b\bigr) \notag                      \\
                 & \hspace{13mm} - p^{\star}D_{z}S(p^{\star}Ax^{\star} + b)\bigr)A\Vbar                      \\
                 & = W^{\top} \Delta_{\perp}(\alpha, \beta, p) A \Vbar.
    \end{align*}
    Applying \Cref{lem:P-norm} to the above results in the desired expression for \(\xipp{p}\), since
    \(
    \norm{W^{\top} \Delta_{\perp}(\alpha, \beta, p) A \Vbar} = \norm{WW^{\top} \Delta_{\perp}(\alpha, \beta, p) A \Vbar\,\Vbar^{\top}} = \norm{P \Delta_{\perp}(\alpha, \beta, p) A Q}.
    \)
\end{proof}

\vspace{3mm}

With these lemmas established, we now state the bounds from \Cref{thm:LSR-bounds-simple} for bifurcations in the Firing Rate model.

\vspace{3mm}

\begin{theorem}\label{thm:firing-bounds}
    Consider the Firing Rate dynamics \eqref{eq:firing} under \Cref{asmp:firing}. For all \(\rpr,\rpp>0\) satisfying
    \begin{multline}
        \Lpr\rpr + \Lpp\rpp  <  \sigma_{\min}\left(J_F^{\star}\right) \rpp \\ - \norm{P D_{z}S(p^{\star} A x^{\star} + b) A x^{\star}} \rpr
    \end{multline}
    where the quantities \(\Lpr,\Lpp\) are defined as
    \begin{align*}
        \Lpr & = \underset{(\alpha,p) \in \Bpr}{\sup} \norm{P \begin{bmatrix} \Delta^{\parallel}_{x}(\alpha,p) AV & \Delta^{\parallel}_{p}(\alpha,p) \end{bmatrix}} \\
        \Lpp & = \underset{\substack{(\alpha,p) \in \Bpr                                                                                                            \\ \beta \in \Bpp}}{\sup} \norm{P \Delta_{\perp}(\alpha,\beta,p) A Q}
    \end{align*}
    with \(\Delta_x^{\parallel}\), \(\Delta_p^{\parallel}\), \(\Delta_{\perp}\) defined in \eqref{eq:firing-delta_para}, \eqref{eq:firing-delta_perp}, there exists a smooth implicit map \(\varphi : \Bpr \to \Bpp\) that satisfies \eqref{eq:LSRmap}.
\end{theorem}
\begin{proof}
    Substituting \eqref{def:M_perp}, \eqref{def:L_para}, \eqref{def:L_perp} from Lemmas \ref{lem:M_perp}, \ref{lem:firing-xi_para}, \ref{lem:firing-xi_perp} into \eqref{eq:LSR-bounds-simple} from \Cref{thm:LSR-bounds-simple} yields the theorem statement.
\end{proof}


\section{Application to Consensus Bifurcations}
\label{sec:example}

In recent years, models of the form \eqref{eq:hopfield}, \eqref{eq:firing} have received attention in the literature on collective decision-making and opinion dynamics. In this context, components of the state vector \(x\) represent scalar opinions of individuals in a social network, which are updated in time according to a nonlinear saturating social rule. Of particular interest is the emergence of collective decision equilibria as bifurcations from the neutral \(x = 0\) state. Such indecision-breaking bifurcations have been established in opinion dynamics models with both Hopfield-like structure \cite{franci2015realization,gray2018multiagent,fontan2017multiequilibria} and firing rate-like structure \cite{bizyaeva2021patterns,bizyaeva2023nonlinear,leonard2024fast}. In this section, we apply the bounds established in \Cref{sec:results} to assess the validity region of the \LSR{}-based analysis of indecision-breaking supercritical pitchfork bifurcations in these nonlinear opinion dynamic models over regular balanced consensus graphs.

To specialize \eqref{eq:hopfield}, \eqref{eq:firing} to nonlinear opinion dynamics (NOD) over regular graphs, we set \(C =  dI_{n}\) with \(d > 0\), \(p = u\), let \(A\) be the adjacency matrix of a \(k-\)regular, undirected and connected simple graph with unit edge weights, \(S =  \tanh\), and consider the unbiased case \(b \equiv \mathbf{0}_{n}\). By the definition of a \(k-\)regular graph, the adjacency matrix \(A\) satisfies \(A \mathbf{1}_{n} = k \mathbf{1}_{n}\). For both models, there is an indecision-breaking bifurcation point \((x^{\star},u^{\star})=\left(\mathbf{0}_{n}, \frac{d}{k}\right)\) at which consensus solutions emerge, see \Cref{fig:consensus-manifold} for schematic. Note from \eqref{def:split} that \(x^{\star}=\mathbf{0}_{n}\implies(\alpha_0,\beta_0)=(0,\mathbf{0}_{n-1})\).

\begin{figure}[ht!]
    \centering
    \includegraphics[width=0.7\linewidth]{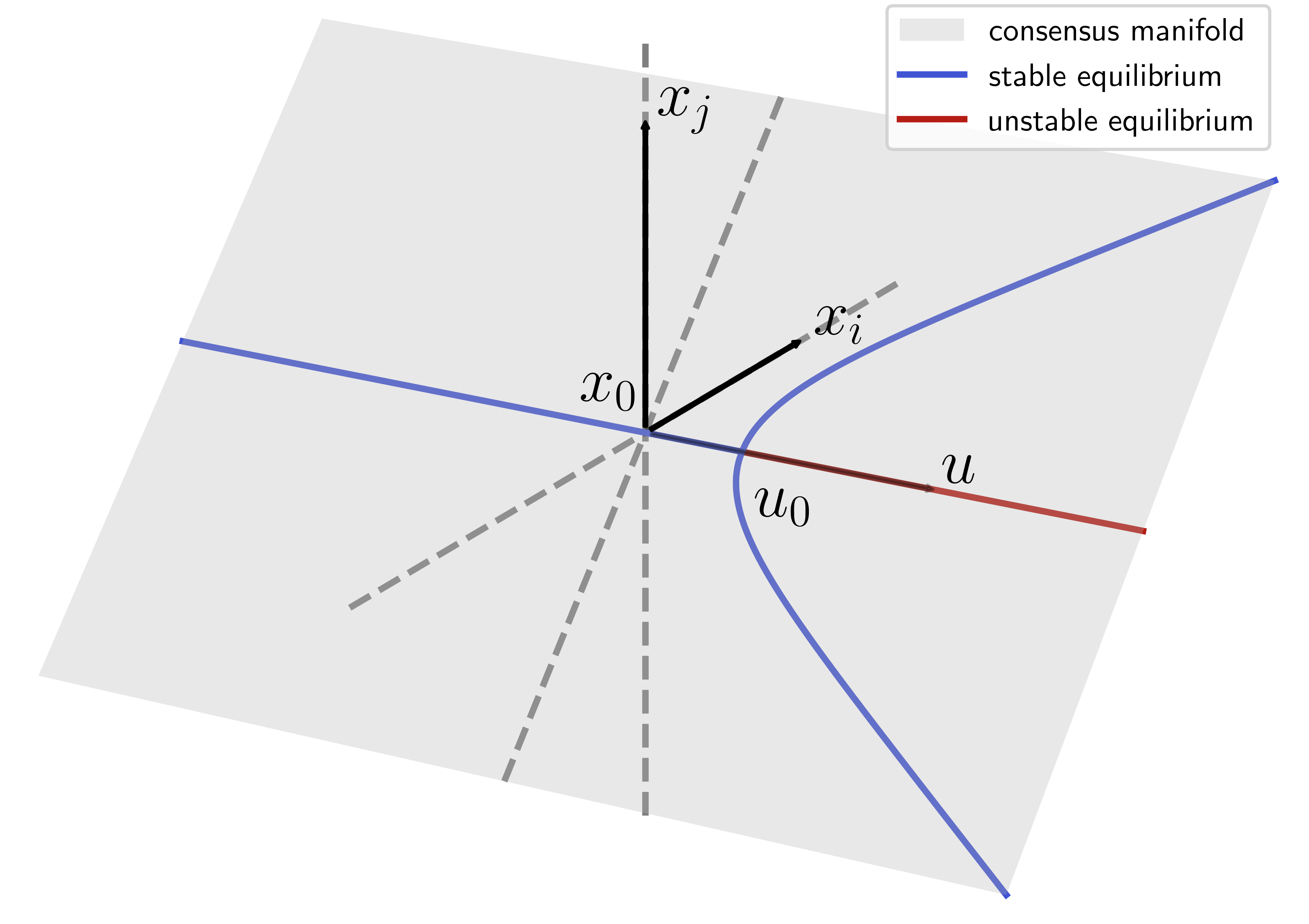}
    \caption{Consensus pitchfork bifurcation in opinion dynamics over a regular graph at \((x^{\star},u^{\star})=\left(\mathbf{0}, \tfrac{d}{k}\right)\); emerging branches of bistable equilibria lie on a manifold tangent to the consensus subspace \(\big\{(x\mathbf{1}_{n},u):x,u\in\R\}\) at the bifurcation point.}
    \label{fig:consensus-manifold}
\end{figure}

\subsection{Hopfield-structured NOD}

Consider \eqref{eq:hopfield} for the \(n-\)node \(k-\)regular consensus graph:
\begin{equation}
    \dot{x} = \Phi(x,u) = -dx + uA\tanh(x). \label{eq:NLOD-hopfield}
\end{equation}

Jacobian \(J^{\star} := D_{x}\Phi(x^{\star},u^{\star})\) of system \eqref{eq:NLOD-hopfield} simplifies to
\begin{equation}
    J^{\star} = -dI_{n}+u^{\star}A\,\mathrm{diag}\big[\sech^2(x^{\star})\big] = \tfrac{d}{k}(A-kI_{n}). \label{eq:jac_Hopfield_NOD}
\end{equation}

We first characterize the orthonormal basis matrix \(V\) of \(\kerJ\). By the Perron--Frobenius theorem \cite[Theorem~8.4.4]{horn2012matrix}, the adjacency matrix \(A\) of a \(k\)-regular, irreducible graph admits a simple largest eigenvalue \(k\), with associated eigenvector \(\mathbf{1}_{n}\). Consequently, \(\kerJ\) is one-dimensional and spanned by the normalized vector \(V = \tfrac{1}{\sqrt{n}}\mathbf{1}_{n}\). Also note that since \(J^{\star}\) is symmetric, we have \(W = \Vbar\) and \(V = \Wbar\), and so \(P = Q\) and \(VV^{\top} + WW^{\top} = I_{n}\).

We now evaluate the quantities defined in \Cref{thm:hopfield-bounds} in order to establish explicit bounds for the bifurcation.
We begin with \(\Mpr\) and \(\Lpr\), both of which evaluate to zero, as we justify below.

From its expression in \eqref{eq:M_para}, we find that
\begin{equation*}
    \Mpr = \tfrac{1}{\abs{u^{\star}}}\norm{P(Cx^{\star} - b)} = \tfrac{k}{d}\norm{(I - VV^{\top})(dI_{n}\mathbf{0}_{n})} = 0.
\end{equation*}
To evaluate \(\Lpr\), we first evaluate the terms \(\Delta_{\parallel}(\alpha,u)\) and \(\bar{S}_{\parallel}(\alpha)\) defined in \eqref{eq:hopfield-xi_para_terms} for \(\xipr{u}\) to obtain
\begin{align*}
    \Delta_{\parallel}(\alpha,u) & = (us_{\alpha} - d/k)I_n,   & \text{with}\;\; s_{\alpha} = \sech^2(\alpha/\sqrt{n}) \\
    \bar{S}_{\parallel}(\alpha)  & = t_{\alpha}\mathbf{1}_{n}, & \text{with}\;\; t_{\alpha} = \tanh(\alpha/\sqrt{n}).
\end{align*}
Substituting these expressions into \Cref{lem:hopfield-xi_para} yields
\begin{align*}
    \norm{\xipr{u}} & = \norm{P A \begin{bmatrix}\left(us_{\alpha} - \tfrac{d}{k}\right) V & t_{\alpha} \mathbf{1}_{n} \end{bmatrix}}  \\
                    & = \norm{Q \begin{bmatrix}\left(us_{\alpha} - \tfrac{d}{k}\right) k V & t_{\alpha} k \mathbf{1}_{n}\end{bmatrix}} \\
                    & = \norm{\mathbf{0}_{(n-1)\times2}} = 0
\end{align*}
since both terms lie in the span of \(V\), and are annihilated by \(Q\).
Thus \(\norm{\xipr{u}} = 0\), and from its definition in \eqref{def:L_para} we conclude that \(\Lpr = 0\).
Consequently, the inequality \eqref{eq:LSR-bounds-simple} reduces to the condition \(\Mpp\Lpp < 1\). We now evaluate \(\Mpp\) and \(\Lpp\) and state the bounds.

To evaluate \(\Mpp\), 
observe that \(\lambda_i(J^{\star}) = \tfrac{d}{k}(\lambda_i(A) - k)\). For a \(k\)-regular, irreducible simple graph with adjacency matrix \(A\), the Perron--Frobenius theorem guarantees that the largest eigenvalue of \(A\) is \(\lambda_1(A) = k\), which is simple, and the subdominant eigenvalue satisfies \(\lambda_2(A) < k\). Consequently,
\begin{equation*}
    \lambda_{1}(J^{\star}) = 0, \quad\text{and}\quad \lambda_{i}(J^{\star}) = \tfrac{d}{k}(\lambda_i(A) - k) < 0, \ i > 1.
\end{equation*}
and hence, the smallest non-zero singular value of \(J^{\star}\) is \(\lambda_{2}(J^{\star}) = \tfrac{d}{k}(\lambda_{2}(A) - k)\) Therefore,
\begin{equation*}
    \Mpp = \Bigl(\tfrac{d}{k}\abs{\lambda_{2}(A) - k}\Bigr)^{-1} = \frac{k}{d\big(k - \lambda_{2}(A)\big)}.
\end{equation*}
Note that the expression for \(\Mpp\) depends inversely on the spectral gap \(k - \lambda_2(A)\). Hence, larger spectral separation between the Perron root and the subdominant eigenvalue of \(A\) leads to smaller values of \(\Mpp\).

Now, to evaluate \(\Lpp\), we first compute the expression for \(\Delta_{\perp}(\alpha, \beta, u)\) defined in \eqref{eq:hopfield-xi_para_terms} for \Cref{lem:hopfield-xi_perp} as
\begin{equation*}
    \Delta_{\perp}(\alpha, \beta, u) = u \diag\bigl[\sech^2\bigl(\Gamma(\alpha,\beta)\bigr)\bigr] - \tfrac{d}{k}I_{n}
\end{equation*}
To evaluate \(\norm{\xipp{u}}\), we first bound it from above using the submultiplicative property of the spectral norm
\begin{align*}
    \norm{\xipp{u}} & = \norm{P A \Delta_{\perp}(\alpha,\beta,u) Q}               \\
                    & \leq \norm{Q A Q} \norm{Q \Delta_{\perp}(\alpha,\beta,u) Q}
\end{align*}
The spectral norm of the projection of \(A\) onto \(\kerJ^{\perp}\) simplifies to \(\max_{i \neq 1}\{\abs{\lambda_{i}(A)}\} = \max\{\abs{\lambda_2(A)}, \abs{\lambda_n(A)}\} =: \abs{\lambda'(A)}\) where \(\lambda_2(A)\) and \(\lambda_n(A)\) denote the subdominant and smallest eigenvalues of \(A\), respectively. For the projection of \(\Delta_{\perp}(\alpha,\beta,u)\) onto \(\kerJ^{\perp}\), we find that its spectral norm is maximized at \((\alpha,\beta,u) = (0, \mathbf{0}_{n-1}, u^{\star} \pm \rpr)\). Since \(\norm{\xipr{u}}\) obtains this upper bound at the same point, we conclude that \(\Lpp = \rpr\,|\lambda'(A)|.\)

Having computed all requisite quantities, we now state the bounds for the Hopfield regular balanced consensus network.
\begin{theorem}\label{thm:NLOD-bounds}
    The bounds of validity of the \LSR{} from \Cref{thm:LSR-bounds} around the pitchfork bifurcation at \((\mathbf{x}^{\star},u^{\star})=(\mathbf{0}, \tfrac{d}{k})\) for the \(n-\)node \(k-\)regular balanced consensus opinion dynamic network \eqref{eq:NLOD-hopfield} is
    \begin{equation*}
        0<\rpr<\frac{d(k-\lambda_{2}(A))}{k\abs{\lambda'(A)}} \quad\text{and}\quad \rpp>0
    \end{equation*}
    along and normal to the consensus manifold \((\kerJ\times\R)\) (\Cref{fig:consensus-manifold}) respectively.
\end{theorem}
\begin{proof}
    Substituting the values of \eqref{def:M_perp} and \(\Mpr = 0\), \(\Lpr = 0\), \(\Mpp = k/(d(k - \lambda_2(A))\), \(\Lpp = \rpr |\lambda'(A)|\) into the inequality \eqref{eq:LSR-bounds-simple} from \Cref{thm:LSR-bounds-simple} yields the desired bounds.
\end{proof}

\begin{figure}[ht!]
    \centering
    \includegraphics[width=\linewidth]{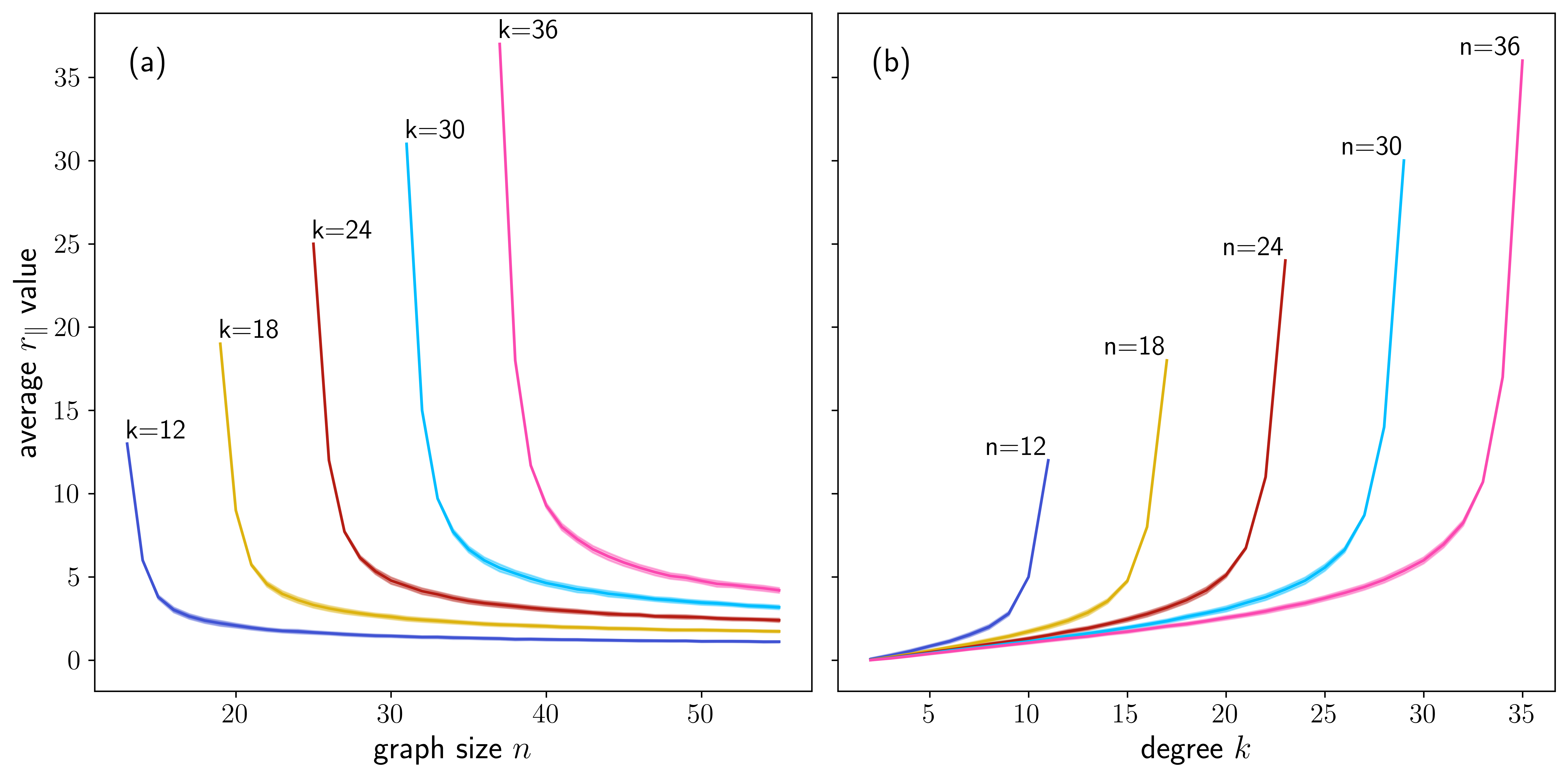}
    \caption{Variation in bounds of validity for \LSR{} for consensus bifurcation over balanced regular graphs.\\
        For each \((n,k)\) pair, 100 independent random graphs were generated using the NetworkX package in Python.
        The reported mean values and corresponding error bars represent the average and standard deviation of the computed bounds on \(\rpr\) across the sampled graphs.
        (a) illustrates the bounds with fixed degree \(k\in\{12,18,24,30,36\}\), varying graph size \(n\);
        (b) illustrates the bounds with fixed graph size \(n\in\{12,18,24,30,36\}\), varying degree \(k\).}
    \label{fig:nk_dependence}
\end{figure}

\vspace{3mm}

\begin{corollary}\label{cor:complete}
    For the case of the complete graph on \(n-\)nodes, i.e. \((n-1)-\)regular, the bounds from \Cref{thm:NLOD-bounds} simplify to \(0<\rpr<n\) and \(\rpp>0\)
\end{corollary}
\begin{proof}
    For the complete graph, the adjacency matrix can be expressed as \(A=\mathbf{1}_{n \times n}-I_{n}\) which has a null diagonal and unit non-diagonal entries, so each vertex has degree \(k=n-1\). The matrix \(\mathbf{1}_{n \times n}=:C\) has \(\lambda_{1}(C)=n\) with corresponding eigenvector \(\mathbf{1}_{n}\), while the remaining eigenvalues \(\lambda_{j\geq2}=0\) with the \(n-1\) eigenvectors orthogonal to \(\mathbf{1}_{n}\). Consequently, \(\lambda_{1}(A)=n-1\) and \(\lambda_{j\geq2}=-1\). Substituting these into \Cref{thm:NLOD-bounds} yields \(0<\rpr\leq n\) and \(\rpp>0\).
\end{proof}

\subsection{Firing-Rate-structured NOD}

We now repeat this analysis for \eqref{eq:firing} over the same \(n-\)node \(k-\)regular consensus graph:
\begin{equation}
    \dot{x} = \Psi(x,u) = -dx + \tanh(uAx). \label{eq:NLOD-firing}
\end{equation}

We will now show that all the relevant quantities admit identical expressions between the Hopfield and Firing-Rate models.
Firstly, Jacobian \(J^{\star} := D_{x}\Psi(x^{\star},u^{\star})\) of system \eqref{eq:NLOD-firing} simplifies to \(J^{\star} = -d I_{n} + \tfrac{d}{k} \diag\bigl[\sech^2(uA\mathbf{0}_{n})\bigl] A = \tfrac{d}{k}(A-kI_{n})\) which is precisely the same expression obtained in the Hopfield case \eqref{eq:jac_Hopfield_NOD}. Consequently, we have the same orthonormal basis \(V = \tfrac{1}{\sqrt{n}}\mathbf{1}_{n}\) for \(\kerJ\).

To show that \(\Mpr = 0\), observe that
\begin{equation*}
    \Mpr = \norm{P \diag\bigl[\sech^2(\tfrac{d}{k}A\mathbf{0}_{n})\bigr]A\mathbf{0}_{n}} = 0
\end{equation*}
and for \(\Lpr\), we first compute the terms defined in \eqref{eq:firing-delta_para} for \(\xipr{u}\) to obtain
\begin{equation*}
    \Delta_{x}^{\parallel}(\alpha,u) = (us_{uk\alpha} - d/k)I_n, \qquad \Delta_{u}^{\parallel}(\alpha,u) = k \alpha s_{uk\alpha} \mathbf{1}_{n},
\end{equation*}
where \(s_{uk\alpha} = \sech^2(uk\alpha/\sqrt{n})\). Substituting into \Cref{lem:firing-xi_para}
\begin{align*}
    \norm{\xipr{u}} & = \norm{P \begin{bmatrix} (u s_{uk\alpha} - d/k) A V & k \alpha s_{uk\alpha} \mathbf{1}_{n} \end{bmatrix}}   \\
                    & = \norm{Q \begin{bmatrix} (u k s_{uk\alpha} - d) V & k \alpha s_{uk\alpha} \mathbf{1}_{n} \end{bmatrix}} = 0
\end{align*}
since both terms lie in the span of \(V\), which is annihilated by \(Q\). From \eqref{def:L_para} we conclude that \(\Lpr = 0\).
Again, inequality \eqref{eq:LSR-bounds-simple} reduces to \(\Mpp \Lpp < 1\). It remains to compute \(\Mpp\) and \(\Lpp\).

Since \(J^{\star}\) remains unchanged from the Hopfield model, we identically have \(\Mpp = \frac{k}{d\big(k - \lambda_{2}(A)\big)}\).
To justify that \(\Lpp\) also remains unchanged, we compute the expression for \(\Delta_{\perp}\) defined in \eqref{eq:firing-delta_perp} for \Cref{lem:firing-xi_perp} to get
\begin{align*}
    \Delta_{\perp}(\alpha,\beta,u) & = u\diag\bigl(\sech^2\bigl(uA\Gamma(\alpha,\beta)\bigr)\bigr)                      \\
                                   & \qquad - \tfrac{d}{k}\diag\bigl(\sech^2(\tfrac{d}{k}A\mathbf{0}_{n})\bigr)         \\
                                   & = u\diag\bigl(\sech^2\bigl(uA\Gamma(\alpha,\beta)\bigr)\bigr) - \tfrac{d}{k}I_{n}.
\end{align*}
With the same reasoning as the Hopfield case, we can say that \(\norm{\xipr{u}}\) attains its supremum at \((\alpha,\beta,u) = (0, \mathbf{0}_{n-1}, u^{\star} \pm \rpr)\).
Consequently, identical to the Hopfield case, \(\Lpp = \rpr\abs{\lambda'(A)}\). Therefore, the bounds of validity for the reduced order bifurcation equations obtained via \LSR{} at the indecision-breaking bifurcation in the Firing Rate-form NOD model \eqref{eq:NLOD-firing} coincide with those of the Hopfield-form NOD model \eqref{eq:NLOD-hopfield}, as stated in \Cref{thm:NLOD-bounds}. We illustrate how these bounds depend on the network size \(n\) and node degree \(k\) in \Cref{fig:nk_dependence}.


\section{Discussion}
\label{sec:conclusion}

In this article we established explicit, computable bounds for the validity of \LSR{}-based bifurcation analysis in representative classes of networked nonlinear systems, and demonstrated their utility by applying the bounds to an illustrative class of nonlinear consensus problems on regular graphs, for which the evaluation proves to be analytically computable.
Notably, the derived expressions depend directly on interpretable features, such as the spectral properties of the underlying graph, providing insight into how network structure influences these bounds.
In future work, we will aim to identify additional classes of systems with bifurcations analysed with the \LSR{} for which these bounds can be computed analytically, and to analyse how these bounds behave across other graph families, such as Erdős-Rényi and small-world networks.


\section{Acknowledgement}

Ravi N Banavar's work was partially supported by a MATRICS grant from the Anusandhan Research Foundation (ANRF).


\bibliographystyle{ieeetr}
\bibliography{ref}

\end{document}